\documentclass[12pt]{article}

\usepackage{epsfig}

\usepackage{amssymb}
\usepackage{amsfonts}

\usepackage{color}
 
\def\be{\begin{equation}}
\def\ee{\end{equation}}
\def\ba{\begin{array}{c}}
\def\ea{\end{array}}

\newcommand{\bea}{\begin{eqnarray}}
\newcommand{\eea}{\end{eqnarray}}

\newcommand{\bbr}{\br\!\br}
\newcommand{\kkt}{\kt\!\kt}

\newcommand{\kt}{\rangle}
\newcommand{\br}{\langle}

\newtheorem{thm}{Theorem}
\newtheorem{cor}[thm]{Corollary}
\newtheorem{lemma}[thm]{Lemma}

\newenvironment{proof}{\noindent
 {\bf Proof.}}{\hfill$\square$\vspace{3mm}\endtrivlist}

\begin{document}

\begin{center}

{\Large \bf

Zig-zag-matrix algebras and solvable quasi-Hermitian quantum models

}

\vspace{0.8cm}

  {\bf Miloslav Znojil}$^{1,2,3}$

\end{center}

%
%
%
%

\vspace{10mm}

 $^{1}$ {The Czech Academy of Sciences,
 Nuclear Physics Institute,
 Hlavn\'{\i} 130,
250 68 \v{R}e\v{z}, Czech Republic, {e-mail: znojil@ujf.cas.cz}}


 $^{2}$ {Department of Physics, Faculty of
Science, University of Hradec Kr\'{a}lov\'{e}, Rokitansk\'{e}ho 62,
50003 Hradec Kr\'{a}lov\'{e},
 Czech Republic}

$^{3}$ Institute of System Science, Durban University of Technology,
Durban, South Africa



\vspace{10mm}

\section*{Abstract}

In quantum mechanics of unitary systems using non-Hermitian (or,
more precisely, $\Theta-$quasi-Hermitian) Hamiltonians $H$ such that
$H^\dagger \Theta=\Theta\,H$, the exactly solvable $M-$level
bound-state models with arbitrary $M\leq \infty$ are rare. A new
class of such models is proposed here, therefore. Its exact
algebraic solvability (involving not only the closed formulae for
wave functions but also the explicit description of all of the
eligible metrics $\Theta$) was achieved due to an extremely sparse
(viz., just $(2M-1)-$parametric) but still nontrivial
``zig-zag-matrix`` choice of the form of $H$.

\section*{Keywords}
.

non-Hermitian quantum mechanics of unitary systems;

a zig-zag-matrix class of $N-$state solvable models;

closed formulae for  wave functions;

closed formula for general physical inner-product metric

\newpage

\section{Introduction}

One of the key obstacles encountered during transition from
classical to quantum mechanics is that the corresponding evolution
equations become operator equations. For this reason the
experimentally testable quantum-theoretical predictions become, in
general, incomparably more difficult. As a consequence, our
understanding of the quantum dynamics becomes only too often
dependent on an analysis mediated by some thoroughly simplified
models of the physical reality {in which, typically,
a given selfadjoint Hamiltonian can be easily diagonalized,
$\mathfrak{h} \to \mathfrak{h}_{diagonal}$}.

{In 1956, Freeman Dyson \cite{Dyson} had to deal with
a fairly complicated multifermionic Hamiltonian $\mathfrak{h}$
for which the convergence
of the conventional numerical
diagonalization algorithms happened to be prohibitively slow.
Still, he managed to find a way out of the difficulty.
His construction of the bound states became nicely convergent
when he preconditioned his Hamiltonian,
 \be
 \mathfrak{h}\ \to \  H = \Omega^{-1}\,\mathfrak{h}\,\Omega\,.
 \label{jednolit}
 \ee
The essence of his convergence-acceleration recipe
lied in a judicious guess of a sufficiently effective
preconditioning (\ref{jednolit}) mediated by a suitable
invertible
mapping $\Omega$. In the language of physics,
this choice just reflected
the role of the correlations in the many-body system
in question. In this sense, the Dyson's}
simplification-oriented
model-building strategy
{found a number of applications,
first of all, in nuclear physics
where the role of the short-range correlations
is fairly well understood
as well as sufficiently easily simulated \cite{Jenssen}.
}

{The
originality of the Dyson's
innovation was that his mappings $\Omega$
were allowed non-unitary,} $\Omega^\dagger\,\Omega \neq I$.
The simplification
(\ref{jednolit}) has been achieved, paradoxically,
at an expense of the loss of the Hermiticity of the Hamiltonian.
{In the language of mathematics,
this can be perceived as} an
unusual, non-unitary transition from a
conventional Hilbert space (say, ${\cal L}$)  to another, auxiliary
but user-friendlier Hilbert space (say, ${\cal H}_{math}$).
In
the language of operators one moves from the conventional textbook
representation of a realistic
Hamiltonian which is self-adjoint in ${\cal L}$,
$\mathfrak{h}=\mathfrak{h}^\dagger$,
to its isospectral
(and, presumably, significantly simpler) manifestly non-Hermitian
avatar  $H \neq H^\dagger$ in ${\cal
H}_{math}$.

In  1992, Scholtz et al \cite{Geyer} proposed a
{different, albeit closely related
model-building} strategy. {These
authors assumed that we are given a non-Hermitian operator $H$
(or rather a set of such operators) in advance.
Under this assumption
they described the way how this operator or operators could
``constitute a consistent quantum mechanical system''.
Thus, in our present notation they just considered
an inverted correspondence (\ref{jednolit}),}
 \be
 H\ \to \   \mathfrak{h}= \Omega\,H\,\Omega^{-1}\,.
 \label{oposm}
 \ee
In {such a} deeply innovative approach one {\em preselects\,} a
suitable tentative non-Hermitian candidate for the Hamiltonian $H
\neq H^\dagger$ {from the very beginning}. Although  {the approach}
has recently been enriched by the development of mathematical
techniques in which the feasibility of practical calculations has
been {enhanced (see, e.g., the more recent review \cite{ali})}, its
mathematical aspects are still full of open questions (see, e.g.,
monograph \cite{book}).

In applications,
naturally, an internal consistency of the theory
based on {reconstruction (\ref{oposm})}
must be guaranteed.
Thus, the spectrum of
$H$ must be real: In this respect it often helps when $H$ is chosen
parity-time symmetric \cite{BB,Carl}. Secondly, many rather
unpleasant emerging mathematical obstacles (see, e.g., their
descriptions in \cite{Trefethen,Siegl,Viola,Uwe}) may be
circumvented when the states of the system in question are
represented in an $M-$dimensional Hilbert space ${\cal
H}_{math}^{(M)}$ where $M$ is arbitrarily large but finite
\cite{Dyson,Geyer}.

Under these conditions (see also \cite{ali} for more details) the
implicit, hidden Hermiticity (or, in mathematics, quasi-Hermiticity
\cite{Dieudonne}) of the operator $H$ representing an input information about
dynamics has to be made explicit.
Once we abbreviate $\Omega^\dagger\,\Omega=\Theta$
(calling this product a ``physical Hilbert-space inner-product
metric''), the standard and conventional textbook self-adjointness
requirement $\mathfrak{h}=\mathfrak{h}^\dagger$ becomes formally
equivalent to the quasi-Hermiticity of $H$ in ${\cal
H}_{math}^{(M)}$,
 \be
 H^\dagger\,\Theta=\Theta\,H\,.
 \label{zasu}
 \ee
{This makes the reconstruction (\ref{oposm})
of $\mathfrak{h}$ redundant.}
In the words of review \cite{Geyer} one manages
{to find a physical inner-product}
metric $\Theta$
compatible
with Eq.~(\ref{zasu})
{``if it exists'' (i.e.,
just in certain parameter regimes)}.

For practical purposes the use of the quasi-Hermitian
formulation of quantum mechanics makes sense only if
Eq.~(\ref{zasu}) as well as the related bound-state Schr\"{o}dinger
equation
 \be
 H^{}\,|\psi_n^{}\kt=E_n^{}\,|\psi_n^{}\kt\,,\ \ \ \ \ n = 1,2, \ldots,
 M
 \label{SE}
 \ee
remain sufficiently user-friendly and solvable. In fact, there exist
not too many solvable models of such a type. One category of the
technical obstacles emerges when $H$ is a differential operator.
Indeed, as long as these operators are, typically, unbounded, the
abstract quantum theory of Ref.~\cite{Geyer} (where all of the
operators of observables have been assumed bounded) cannot be
applied.

Even when both of our above-mentioned Hilbert spaces ${\cal L}$ and
${\cal H}_{math}^{(M)}$ are kept finite-dimensional, $M < \infty$,
the literature offers just a few toy-matrix models $H$ which remain
exactly solvable, {at an arbitrary number of states
$M < \infty$,} in the manner which combines the availability of a
closed form of all of the solutions $|\psi_n\kt$ and $E_n$ of
Schr\"{o}dinger Eq.~(\ref{SE}) with the equally important
availability of a closed form of at least one of the solutions
$\Theta=\Theta(H)$ of Eq.~(\ref{zasu}).

{In these models (see, e.g.,
\cite{maximal,tridiagonal} or \cite{procA}, with further references)
one still has to work with the tridiagonal forms of the
Hamiltonians. In what follows we intend to propose the class of
solvable models in which the Hamiltonians form even a sparse-matrix
subset of the similar tridiagonal models. They will form} a new
exactly solvable family of unitary quasi-Hermitian quantum models.
{We will see that these models can be perceived as an
illustration of the situation in which the unitary quantum model
based on a manifestly non-Hermitian Hamiltonian $H\neq H^\dagger$
appears preferable and, not quite expectedly, technically simpler
than its isospectral Hermitian-matrix alternative of conventional
textbooks.}

\section{Exact solution of Schr\"{o}dinger equation}

It is not too surprising that in the majority of the realistic
applications of the bound-state Schr\"{o}dinger equations using a
self-adjoint phenomenological Hamiltonian $\mathfrak{h}$ people
recall the variational argument and approximations and keep the
dimension $M$ of the conventional textbook Hilbert space ${\cal L}$
finite {\cite{Dyson,Jenssen,Geyer}}.
Then, there are also no conceptual problems with the
linear-algebraic correspondence between ${\cal L}$ and ${\cal
H}_{math}^{(M)}$ and/or between $\mathfrak{h}$ and $H$ (cf.
Eq.~(\ref{jednolit})).

{The situation is different when the Hamiltonians
$\mathfrak{h}$ and/or $H$ are differential operators with $M=\infty$.
On positive side, the standard ```kinetic plus potential energy''
structure of such a class of operators makes them intuitively
acceptable on
physical
grounds: Typically,
this renders them  eligible in the role of prototype models
in quantum field theory \cite{Carl}.
For this reason, even on the level of quantum mechanics
the dedicated literature abounds with the exactly solvable models
\cite{ES} as well as with the quasi-exactly solvable models
\cite{QES} - \cite{QES3}
of such a type.}

{On negative side, the recent progress
in the
analysis of the $\mathfrak{h}\ \leftrightarrow \ H$ correspondence
led to several disappointing
disproofs of its existence \cite{book}. {\it Pars pro toto\,}
it is sufficient to mention papers \cite{Siegl,Viola}
containing the mathematically rigorous
disproofs of the existence of {\em any\,} self-adjoint
partner $\mathfrak{h}$ for the most popular
imaginary cubic oscillator Hamiltonian $H$ of Ref.~\cite{BB}.}

{After all, the very explicit words of warning were
already written in the older review \cite{Geyer}.
The authors required there that {\em any\,}
eligible non-Hermitian operator representing an observable
should be bounded.}
In other words,
{under the warmly recommended
auxiliary assumption $M<\infty$
the mathematics becomes perceivably simpler.}
The problems which remain to
be resolved
are purely technical, emerging usually just at
sufficiently large matrix dimensions $M \gg 1$ and requiring only
a sufficiently reliable numerical software.

The prevailing nature of results is then purely numerical.
The
exactly solvable bound-state models are rare. Even the
diagonalization of a next-to-diagonal (i.e., tridiagonal) matrix
form of $\mathfrak{h}$ may be ill-conditioned and just badly
convergent \cite{Wilkinson}. In this context the guiding mathematical
idea of our present project
was that one of the rarely emphasized consequences of the choice of
a non-Hermitian model $H$ with real spectrum is that its nontrivial
(i.e., non-diagonal) matrix representation can be ``sparse tridiagonal''.

An exciting formal appeal of the latter idea appeared accompanied
by the emerging possibility of its transfer to the
phenomenology and
physics of various lattice models \cite{Deguchi}.
Both of these observations led us directly to the introduction and study of
the $M$ by $M$ ``zig-zag-matrix'' (ZZM) Hamiltonians
 \be
 H=H^{(ZZM)}(\vec{a},\vec{c})=
 \left[ \begin {array}{cccccc}
     {\it a_1}&0&0&0&\ldots&
 \\\noalign{\medskip}{\it c_1}&{\it a_2}&{\it c_2}&0&\ldots&
 \\\noalign{\medskip}0&0&{\it a_3}&0&0&\ldots
 \\\noalign{\medskip}0&0&{\it c_3}&{\it a_4}&{\it c_4}&\ddots
 \\\noalign{\medskip}\vdots&\vdots&0&0&{\it a_5}&\ddots
 \\\noalign{\medskip}&&\vdots&\ddots&\ddots&\ddots
 \end {array} \right]\,
 \label{defzz}
 \ee
in which just $2M-1$ real parameters do not vanish.

A compact outline of some of the purely mathematical properties of
matrices (\ref{defzz}) may be found postponed to Appendix A below.
The bound-state spectrum of these matrices (i.e., of the
Hamiltonians of our present interest) coincides with the subset of
parameters $a_1, a_2, \ldots, a_M$ occupying the main diagonal (see
Lemma \ref{lemmajedna} in the Appendix). This means that the
unitarity of the evolution is guaranteed by the reality
of the spectrum of energies, i.e., by the reality
of these dynamical-input
parameters.

For the purposes of applications we are just left with the necessity
of the construction of the wave functions {i.e., in
the conventional Dirac's notation, of the column-vector solutions
$|\psi_1^{}\kt$, $|\psi_2^{}\kt$, \ldots of our Schr\"{o}dinger
Eq.~(\ref{SE}) corresponding to the respective bound-state energies
$E=a_n$ with $n=1,2,\ldots$. Surprisingly enough, these ket-vectors
can be obtained} in closed form. Incidentally, the construction is
most straightforward when $M=\infty$ {because in such
a case we do not need to separate the description of the solutions
at even and odd $M<\infty$.}

{Another useful trick used during the explicit
systematic construction of the solutions of our Schr\"{o}dinger
Eq.~(\ref{SE}) is that at any $M \leq \infty$ we can concatenate our
ket-vector columns into a single $M$ by $M$ matrix, say
 \be
 \{
 |\psi_1^{}\kt,|\psi_2^{}\kt,\ldots,|\psi_M^{}\kt
 \}=Q_{solution}\,.
 \label{matfor}
 \ee
Indeed, precisely the study of this matrix-of-solutions leads to the
following important result.}

\begin{lemma}
\label{lemmaacko} In Schr\"{o}dinger Eq.~(\ref{SE}) with $M=\infty$
and with the ZZM Hamiltonian $H=H^{(ZZM)}(\vec{a},\vec{c})$, the
column-vector eigenstates $|\psi_n^{}\kt$ corresponding to the
energies $E=a_n$ with $n=1,2,\ldots$ {and arranged in
matrix (\ref{matfor}) acquire precisely the ZZM-matrix form defined
in terms of suitable vectors of parameters
$\vec{x}=\{x_1,x_2,\ldots\}$ and $\vec{y}=\{y_1,y_2,\ldots\}$,
 \be
 Q_{solution}
 =H^{(ZZM)}(\vec{x},\vec{y})\,.
 \label{normacko}
 \ee}
Under the auxiliary {\it ad hoc\,} assumption that $c_j\neq 0$ at
all odd $j$ we may accept, {say, the following}
normalization {of the separate ket-vector columns of
$Q_{solution}$},
 \be
 x_2=x_4=\ldots = 1\,,\ \ \ \ y_1=y_3=\ldots = 1\,.
 \label{normbecko}
 \ee
Then, the {closed-form} solution of our
infinite-dimensional matrix Schr\"{o}dinger equation is given by
formulae
 \be
 x_j=(a_j-a_{j+1})/c_j\,,\ \ \ \ \ j = {\rm odd}
 \ee
and
 \be
 y_k=-\frac{(a_{k+1}-a_{k+2})c_k}{(a_{k}-a_{k+1})c_{k+1}}\,,\ \ \ \ \ k = {\rm
 even}\,.
 \ee
\end{lemma}
\begin{proof}
Proof {is based on the auxiliary lemmas of Appendix A
reflecting the remarkable properties of the algebra of zig-zag
matrices. The formulae themselves} follow directly from the
insertion of the solution in Schr\"{o}dinger equation.
\end{proof}

In this Lemma
our assumption $M=\infty$ enabled us to avoid the
discussion of the role of the truncation
of the matrix at $M < \infty$. In the latter case,
fortunately, it proves sufficient to
set, formally, $a_{M+1}=a_{M+2}=\ldots =0$
and $c_M=c_{M+1}=\ldots =0$.
Also
the apparent $c_j \to 0$
singularities at $j<M$
are just an artifact of our normalization
(\ref{normbecko}). Whenever needed, these singularities may be
removed easily because
our choice of the normalization has been
dictated by the simplicity of the proof rather than by the
simplicity or optimality of the formulae (\ref{normacko}) and (\ref{normbecko}).
The amendment is offered by the following re-normalized and more compact
result.

\begin{lemma}
\label{lemmaackobe} In Schr\"{o}dinger Eq.~(\ref{SE}) with the ZZM
Hamiltonian $H=H^{(ZZM)}(\vec{a},\vec{c})$, the column-vector
eigenstates corresponding to the bound-state
energies $E=a_n$ with $n=1,2,\ldots,M$ can be given a differently
normalized {``tilded''} form
 \be
 \{
 |\widetilde{\psi_1^{}}\kt,|\widetilde{\psi_2^{}}\kt,
 \ldots,|\widetilde{\psi_M^{}}\kt
 \}
 =H^{(ZZM)}(\vec{p},\vec{q})
 \label{normackocis}
 \ee
where we {employ a different,} unit-diagonal
normalization $p_j=1$ at all $j$, and where we obtain the more
compact {formula for the off-diagonal parameters
forming the vector $\vec{q}$,}
 \be
 q_k=-c_k/(a_k-a_{k+1})\,,\ \ \ \ \ k = 1,2,\ldots,M-1\,.
 \ee
\end{lemma}
\begin{proof}
{As long as we just changed the normalization
convention, there exists a diagonal matrix (say, $\varrho$) such
that $H^{(ZZM)}(\vec{x},\vec{y})\,\varrho =
H^{(ZZM)}(\vec{p},\vec{q})$}.
\end{proof}

\section{Closed-form construction of all of the eligible metrics}

It is well known \cite{SIGMAdva} that whenever we replace the
manifestly non-Hermitian Hamiltonian $H$ in Schr\"{o}dinger
Eq.~(\ref{SE}) by its conjugate $H^\dagger$, the knowledge of the
``ketket'' solutions of the associated Schr\"{o}dinger equation
 \be
 H^\dagger\,|\psi_n^{}\kkt=E_n^{}\,|\psi_n^{}\kkt\,,
 \ \ \ \ n=1,2,\ldots,M
 \label{SEer}
 \ee
enables us to define all of the admissible metrics
$\Theta=\Theta(H)$ (i.e., all of the admissible solutions of
Eq.~(\ref{zasu})) by formula
 \be
 \Theta=\Theta(\kappa^2_1\,,\kappa^2_2\,, \ldots,\kappa^2_M)
 =\sum_{n=1}^M\,|\psi_n^{}\kkt\,\kappa^2_n\,\bbr \psi_n^{}|\,.
 \label{odkap}
 \ee
This is {\em not\,} a spectral representation of $\Theta$ because in
general (i.e., due to the non-Hermiticity of $H$) the overlaps $\bbr
\psi_m^{}| \psi_n^{}\kkt$ need not vanish even when $m \neq n$.
Still, this formula shows that the general metric can vary with as
many as $M$ freely variable real and positive parameter
$\kappa^2_n$.

In comparison with Eq.~(\ref{SE}),
the most important comment
concerning Eq.~(\ref{SEer}) is that as long as our toy-model ZZM
Hamiltonians $H$ are real, we now have to deal with the transposed
matrices,
 \be
 H^\dagger=H^T=H^{(TZZM)}(\vec{a},\vec{c})=
 \left[ \begin {array}{cccccc}
     {\it a_1}&{\it c_1}&0&0&\ldots&
 \\\noalign{\medskip}0&{\it a_2}&0&0&\ldots&
 \\\noalign{\medskip}0&{\it c_2}&{\it a_3}&{\it c_3}&0&\ldots
 \\\noalign{\medskip}0&0&0&{\it a_4}&0&\ddots
 \\\noalign{\medskip}\vdots&\vdots&0&{\it c_4}&{\it a_5}&\ddots
 \\\noalign{\medskip}&&\vdots&\ddots&\ddots&\ddots
 \end {array} \right]\,.
 \label{Tdefzz}
 \ee
The crucial consequence is that it is sufficient to replace the
ZZM theory of Appendix A by its transposed-matrix TZZM alternative.

\begin{lemma}
\label{lemmaTackobe} In Schr\"{o}dinger Eq.~(\ref{SEer}) with the
transposed Hamiltonian $H^T=H^{(TZZM)}(\vec{a},\vec{c})$, the
collection of the column-vector eigenstates $|\psi_n^{}\kkt$
corresponding to the bound-state energies $E=a_n$ with
$n=1,2,\ldots,M$ can be given the TZZM form,
 \be
 \{
 |\psi_1^{}\kkt,|\psi_2^{}\kkt,\ldots,|\psi_M^{}\kkt
 \}
 =H^{(TZZM)}(\vec{p},\vec{q})\,.
 \label{normTackocis}
 \ee
The normalization $p_j=1$ (at all $j$) leads to the
closed-form result
 \be
 q_k=-c_k/(a_k-a_{k+1})\,,\ \ \ \ \ k = 1,2,\ldots,M-1\,.
 \label{defel}
 \ee
\end{lemma}
\begin{proof}
Proof is a TZZM analogue of the ZZM proof of Lemma
\ref{lemmaackobe}.
\end{proof}

Formula (\ref{normTackocis}) containing $M-1$
characteristics (\ref{defel}) of the Hamiltonian may be inserted in
the definition of all of the eligible metrics (\ref{odkap}). The
resulting $M$ by $M$ matrices $\Theta$ would be, by construction,
invertible, Hermitian and positive definite.
Due to the reality and tridiagonality of the factor
(\ref{normTackocis}) and of its transposition, all of the metrics
will have a real and symmetric pentadiagonal-matrix form.
The explicit evaluation of their matrix elements is
straightforward
and constitutes our present main mathematical result.


\begin{thm}
\label{theorema} Every metric $\Theta$ guaranteeing the
quasi-Hermiticity (\ref{zasu}) of our $(2M-1)-$parametric ZZM
Hamiltonian (\ref{defzz}) can be given the three-component form
 \be
 \Theta= \Theta^{(diag)}+ \Theta^{(tridiag)}+
 \Theta^{(pentadiag)}\,.
 \label{Tdefzz}
 \ee
Its first component is just the invertible, $q_j-$independent and
positive-definite diagonal matrix,
 \be
  \Theta^{(diag)}=H^{(TZZM)}(\vec{\kappa^2},\vec{0})\,.
 \ee
The second component has the sparse tridiagonal-matrix form with
vanishing main diagonal,
 \be
  \Theta^{(tridiag)}_{nn}=0\,,\ \ \ \ n=1,2,\ldots,M\,.
 \ee
Its off-diagonal elements
 \be
  \Theta^{(tridiag)}_{m,m+1}=\Theta^{(tridiag)}_{m+1,m}=q_m\kappa^2_{m+1}\,,
  \ \ \
  m=1,3,\ldots\ (\leq M-1)
  \ee
and
 \be
  \Theta^{(tridiag)}_{n,n+1}=\Theta^{(tridiag)}_{n+1,n}=q_n \kappa^2_{n}
 \,,\ \ \ \ n=2,4,\ldots\ (\leq M-1)\,
 \ee
are all linear in $q_j$s. The remaining, third component of the
metric has the pentadiagonal sparse-matrix form
 \be
  \Theta^{(pentadiag)}=
 \left[ \begin {array}{cccccccc}
  {{\it q_1}}^{2}{\it \kappa^2_2}&0&{\it q_1}\,{\it \kappa^2_2}\,{\it q_2}&0&\ldots&&&
  \\\noalign{\medskip}0&0&0&0&0&\ldots&&
 \\\noalign{\medskip}{\it q_1}\,{ \it \kappa^2_2}\,{\it q_2}&0&{{\it q_2}}^{2}{\it \kappa^2_2}+{{\it
 q_3}} ^{2}{\it \kappa^2_4}&0&{\it q_3}\,{\it \kappa^2_4}\,{\it
 q_4}&0&\ldots&
 \\\noalign{\medskip}0&0&0&0&0&0&0&\ldots
 \\\noalign{\medskip}\vdots&0&{\it q_3}\,{\it \kappa^2_4}\,{\it q_4}&
 &{{\it q_4}}^{2}{\it \kappa^2_4}+{{\it q_5}}^{2}{\it \kappa^2_6}
 &0&{ \it q_5}\,{\it \kappa^2_6}\,{\it q_6}&\ddots
 \\\noalign{\medskip}&\vdots&0&0&0&0&0&\ddots
 \\\noalign{\medskip}&&\vdots&0&{\it
 q_5}\,{\it \kappa^2_6} \,{\it q_6}&0&{{\it q_6}}^{2}{\it
 \kappa^2_6}+{{\it q_7}}^{2}{\it
 \kappa^2_8}&\ddots
 \\\noalign{\medskip}&&&\vdots&\ddots&\ddots&\ddots&\ddots
 \end {array}
 \right]\,
 \ee
with elements which are all quadratic in $q_j$s.
\end{thm}
\begin{proof}
The result follows directly from formulae (\ref{odkap}) and
(\ref{normTackocis}).
\end{proof}

In the context of physics the latter result is truly
remarkable because it implies that it really does make sense to work
with the present non-Hermitian ZZM or TZZM representations $H$ of
the Hamiltonians. For the simplest dynamical scenarios with
the finite and not too large numbers of the bound-state levels
$M<\infty$ at least, one could feel tempted to employ the standard
algorithms of linear algebra and to factorize the
pentadiagonal-matrix metric $\Theta=\Omega^\dagger\Omega$ of
Theorem \ref{theorema}. In principle, this would yield the
explicit Dyson map $\Omega$ and, finally, enable us to return to the
quantum mechanics of textbooks in which the conventional
self-adjoint representation $\mathfrak{h}$ of the Hamiltonian would
be ``easily'' reconstructed via Eq.~(\ref{oposm}). Nevertheless,
{a {\em feasible} realization of such an alternative,
more traditional version
of the present models would require an invention of new methods.}

{Indeed, the first mathematical
obstacle would emerge when we imagine that the
metrics $\Theta$ of
Eq.~(\ref{odkap}) and of
Theorem \ref{theorema} are ambiguous, $M-$parametric \cite{Geyer}.
Secondly, we would have to
deduce, from factorization $\Theta=\Omega^\dagger\,\Omega$,
a suitable sample of the Dyson map $\Omega$.
Then, indeed, a new set of free parameters forming a
unitary matrix ${\cal U}$ would have to be introduced and
considered here
due to the ambiguity of the factorization of the metric itself,
$\Theta=\Omega^\dagger\,\Omega=\Omega^\dagger\,{\cal U}^\dagger\,{\cal U}\,\Omega$.
Thus, certainly,}
the ``conventional''
Hermitian matrix $\mathfrak{h}$ would be a non-sparse,
user-unfriendly matrix {in general}.
Hence, the non-Hermitian matrix $H$ {really seems to offer}
the most economical representation of the Hamiltonian. One could
hardly find reasons for a tedious reconstruction of its
partner(s) $\mathfrak{h}$ of conventional textbooks.

\section{Conclusions}

In the dedicated literature, not too many quasi-Hermitian quantum
models have the ``exact and complete solvability'' property of our
present class of $M-$level bound-state systems using the real
zig-zag-matrix Hamiltonians (\ref{defzz}) with $2M-1$ free
parameters. Typically, the algebraically solvable models of such a
type are based on the use of tridiagonal matrix forms of $H$ (cf.,
e.g., a sample of such a class of quasi-Hermitian models in
\cite{procA}). In general, given a realistic non-Hermitian $H$, the
metric $\Theta$ assigned to the model is usually just approximate
and not too flexible, corresponding usually just to a fixed choice
of the set of parameters $\kappa_n^2$ in (\ref{odkap}).

In comparison, our present model is rather exceptional in keeping
the whole set of the metric-determining parameters
$\kappa_n^2$ freely variable. Moreover, our
restriction of the class of the Hamiltonians to the mere sparse
zig-zag matrices
of Appendix A proved
fortunate:
We discovered that the full sets of the eigenstates of
$H$ appeared to belong to the same
(viz., ZZM) subclass of the highly sparse zig-zag matrices.
One would even like to say ``serendipitiously fortunate''
because the same comment appeared to apply also to
the TZZM subclass and to the transposed Hamiltonian $H^\dagger$
playing a key role in the construction of the {\em complete\,}
set of the eligible metrics
$\Theta=\Theta(H)$.

In the context of physics one of the remarkable properties of the
model is that its bound-state energy spectrum coincides, due
to the ZZM sparsity of the Hamiltonian, with its main diagonal.
As an input information about dynamics
it can
be, therefore, fixed in
advance. This means that the remaining $M-1$ freely variable
off-diagonal matrix elements of $H$ can be interpreted as playing
an energy-complementing role of parameters responsible for the
operator metric $\Theta$, i.e., for
the correct
physical geometry
of the Hilbert space.
In this manner these parameters
influence, directly and implicitly, the selection and form of
the other possible
observable features of the system \cite{Geyer}.

A final complementary comment may be also added on the existence and
structure of the exactly solvable quantum models of unitary systems
occurring and widely used within the framework of the conventional
Hermitian quantum mechanics in which the metric is kept trivial,
$\Theta_{conventional}=I$. Indeed, once the operators of the
observables (including the Hamiltonians) become required,
in the conventional textbook spirit,
self-adjoint, the first nontrivial matrix form of an observable (or
of the
Hamiltonian) has to be real and symmetric, i.e., fully tridiagonal,
i.e., from the numerical-manipulation perspective, perceivably more
complicated than our ``maximally sparse'' ZZM models (\ref{defzz}).

\newpage

\section*{Appendix A. Zig-zag matrices}

By ``zig-zag matrices'' (ZZM) we will understand, in this paper, the
real and tridiagonal $M$ by $M$ matrices of Eq.~(\ref{defzz}) which
may be finite- or infinite-dimensional (i.e., $M \leq \infty$).
Whenever needed, the transpositions of these matrices will be
called, for the sake of definiteness, the  ``transposed zig-zag
matrices'' (TZZM). Both of these classes of matrices have a few
truly remarkable properties.

\begin{lemma}
\label{lemmajedna} The spectrum of $H^{(ZZM)}(\vec{a},\vec{c})$
coincides with the $M-$plet of parameters $\vec{a}$.
\end{lemma}
\begin{proof}
It is sufficient to recall the definition of the secular determinant
$\det (H^{(ZZM)}(\vec{a},\vec{c})-E)$.
\end{proof}

\begin{lemma}
\label{lemmadve} The standard matrix product of two ZZM factors
retains the ZZM property,
 \be
  H^{(ZZM)}(\vec{a},\vec{c}) \cdot
 H^{(ZZM)}(\vec{b},\vec{d})=H^{(ZZM)}(\vec{u},\vec{v})
 \ee
  with
 \be
 u_j=a_jb_j\,,\ \ \ \ j=1,2,\ldots,M
 \ee
 \be
 v_k=c_k b_k+ a_{k+1}d_k\,,\ \ \ \ k={\rm odd}\,,\ \ k \leq M-1
 \ee
 \be
 v_k=a_k d_k+c_k b_{k+1}\,,\ \ \ \ k={\rm even}\,,\ \ k \leq M-1\,.
 \ee
\end{lemma}
\begin{proof}
Proof is obtained directly from the definition of the standard
matrix product.
\end{proof}

\begin{cor}
Any positive integer power of ZZM retains the ZZM property.
\end{cor}

\begin{lemma}
The negative integer powers of $H^{(ZZM)}(\vec{a},\vec{c})$ retain
the ZZM property, provided only that the zero does not belong to the
spectrum, $a_j \neq 0$, $j=1,2,\ldots, M$.
\end{lemma}
\begin{proof}
It is sufficient to set the product $H^{(ZZM)}(\vec{u},\vec{v})$ in
Lemma \ref{lemmadve} equal to the unit matrix. Then we may check
that in this case (i.e., with all ${u}_j=1$ and with all $v_j=0$)
one can define the inverted matrix
$H^{(ZZM)}(\vec{b},\vec{d})=[H^{(ZZM)}(\vec{a},\vec{c})]^{-1}$ via
the following elementary formulae,
 \be
 b_j=1/a_j\,,\ \ \ \ j=1,2,\ldots,M\,,\ \ \ \ \
 d_k=-c_k /(a_k a_{k+1})\,,\ \ \ \  k =1,2,\ldots, M-1\,,
 \ee
i.e., by the formulae independent of the parity of $k$.
\end{proof}

 \noindent
In addition to these observations it is also easy to check, by the
explicit matrix multiplication, that
 \be
 H^{(ZZM)}(\vec{a},\vec{0})\,[H^{(ZZM)}(\vec{a},\vec{c})]^{-1}\,H^{(ZZM)}(\vec{a},\vec{0})=
 H^{(ZZM)}(\vec{a},-\vec{c})\,.
 \label{indefzz}
 \ee
Moreover, depending on the specific needs in applications it is also
not too difficult to deduce multiple other auxiliary formulae in
which the fractions are eliminated, say, from the formulae for the higher
negative powers of the matrices. For illustration let us only
display here their first nontrivial sample,
 \be
 [H^{(ZZM)}(\vec{a},\vec{0})]^2\,
 [H^{(ZZM)}(\vec{a},\vec{c})]^{-2}\,[H^{(ZZM)}(\vec{a},\vec{0})]^2=
 H^{(ZZM)}(\vec{\widetilde{a}},-\vec{\widetilde{c}})\,
 \ee
where
 $
 \widetilde{a}_j=a_j^2\,,\ \ \ j = 1, 2, \ldots, M$ and $
 \widetilde{c}_k=c_k\,(a_k+a_{k+1})\,,\ \ \ k = 1, 2, \ldots, M-1
 $.

\end{document}